\documentclass[pre,10pt,twocolumn,superscriptaddress,longbibliography,showpacs]{revtex4-1}
\usepackage{amsmath}
\usepackage{amsthm}
\usepackage{amsfonts}
\usepackage{amssymb}
\usepackage{bbm}
\usepackage{graphicx}
\usepackage{mathrsfs}
\usepackage{hyperref} 
\hypersetup{pdfborder={0 0 0},colorlinks=true}

\newtheorem{lemma}{Lemma}
\newtheorem*{lemma*}{Lemma}

\newtheorem*{proposition*}{Proposition}
\newtheorem{theorem}{Theorem}

\newtheorem*{conjecture*}{Conjecture}

\newtheorem*{corollary*}{Corollary}
\theoremstyle{remark}
\newtheorem{example}{Example}

\newcommand{\ket}[1]{\vert #1 \rangle}
\newcommand{\bra}[1]{\langle #1 \vert}

\newcommand{\braket}[2]{\langle #1 \vert #2 \rangle}
\newcommand{\ketbra}[2]{\ket{#1}\bra{#2}}

\newcommand{\abs}[1]{\left| #1\right|}

\newcommand{\Tr}{\operatorname{Tr}}

\newcommand{\Span}{\operatorname{span}}

\newcommand{\hi}{\mathcal{H}} 
\newcommand{\tr}[1]{\mathrm{tr}\left[#1\right]} 

\newcommand{\M}{\mathcal{M}}

\renewcommand{\H}{\mathcal{H}}
\newcommand{\II}{\mathbb{I}}

\begin{document}
\title{On the pure steered states of Einstein-Podolsky-Rosen steering}

\author{H. Chau Nguyen}\email{chau@pks.mpg.de}
\affiliation{Max-Planck-Institut f\"{u}r Physik Komplexer Systeme,~N{\"{o}}thnitzer Stra{\ss}e 38, D-01187 Dresden, Germany}

\author{Kimmo Luoma}\email{kimmo.luoma@tu-dresden.de}
\affiliation{Institut f{\"u}r Theoretische Physik, Technische Universit{\"a}t Dresden, 
D-01062 Dresden, Germany}

\begin{abstract}
In the Einstein--Podolsky--Rosen experiment, when Alice makes a
measurement on her part of a bipartite system, Bob's part is collapsed
to, or steered to, a specific ensemble.  Moreover, by reading her
measurement outcome, Alice can specify which state in the ensemble
Bob's system is steered to and with which probability. The possible states that Alice can steer Bob's system to are called
steered states. In this work, we study the subset of steered states which are pure after normalisation. We illustrate that these
pure steered states, if they exist, often carry interesting
information about the shared bipartite state. 
This information content becomes particularly clear when we study the purification of the shared state. Some applications are discussed. These include a generalisation
of the fundamental lemma in the so-called `all-versus-nothing proof of
steerability' for systems of arbitrary dimension.
\end{abstract}


\maketitle

\subsection*{Introduction}

Einstein--Podolsky--Rosen (EPR) steering, or just steering, 
indicates the ability of Alice to collapse Bob's system to a specific
ensemble of states by measuring locally her own system, provided the
two systems are in an appropriate quantum state. This 
was first discussed in the seminal papers
of Einstein,~Podolsky and Rosen~\cite{einstein_can_1935} and 
Schr\"odinger~\cite{schrodinger_discussion_1935}. Recently, Wiseman~\textit{et al.} gave a precise operational definition of
steerability~\cite{wiseman_steering_2007,
jones_entanglement_2007}, which
distinguishes it from other classes of nonlocality in quantum
mechanics such as nonseparability~\cite{werner_quantum_1989} and
Bell nonlocality~\cite{bell_einstein-podolsky-rosen_1964}.

EPR steering can be seen also as a 
two-party quantum information processing task where 
Alice tries to convince Bob that the state they 
share is entangled. Bob does not trust Alice but he trusts on 
quantum mechanics and his own measurements~\cite{jones_entanglement_2007}. 
This view has motivated a lot of research on how to witness steerability in a specific experimental setting
by inequalities~\cite{cavalcanti_experimental_2009,
cavalcanti_unified_2011,schneeloch_einstein-podolsky-rosen_2013,costa_quantification_2016,zhu_universal_2016,saunders_experimental_2010,schneeloch_violation_2013}. Beyond sufficient inequalities, a sufficient and
necessary condition for steerability in terms of quantum channels has
been developed~\cite{piani_channel_2015,piani_necessary_2015}. It has also been shown that the steerability is related
to the concept of joint measurability of quantum observables 
\cite{uola_joint_2014,quintino_joint_2014,uola_one--one_2015,
uola_adaptive_2016}.

On the other hand, some {authors} concentrate on characterising
the steering process in EPR experiments independently from a specific
measurement setting. In this line, the EPR steering is characterised
by a positive linear map, which maps a measurement outcome operator, or \emph{effect}, on Alice's
system to a conditional state on Bob's
system, which is normalised to a steered state~\cite{nguyen2016non}. 
Then an analysis of an EPR experiment is based 
on all effects on Alice's system, and all
respective conditional states on Bob's
system~\cite{jevtic_quantum_2014,nguyen2016non,jevtic2015einstein,nguyen2016necessary}. The set of conditional states on Bob's system and its relative have been proved useful in characterising not only steerability but also more general aspects
of bipartite quantum states~\cite{jevtic_quantum_2014,nguyen2016non,milne2014quantum,milne2015geometric}.

In this work we concentrate on a special case of steered states: those that
are pure, or rank-$1$. In general, there may be no such pure
steered states. However, we show that if they exist, they often
carry interesting information about the shared bipartite state. The special role
of pure steered states becomes particularly clear when we study the
behavior of the purified system in the EPR experiment. The following
simple observation illustrates the idea. Suppose Charlie holds a third
system such that the joint system of his with Alice and Bob is in a
pure state. When Alice makes a measurement and obtains some outcome,
the {shared system of Bob and Charlie (BC)} is steered to some state, which is
nonseparable in general. However, if Alice learns that Bob's system is
in a pure state, she knows immediately the two {(BC)} are separable and
Charlie's system is also in a pure state. This is a clue that pure steered states imply special structure of the shared state, 
{which becomes visible in}
its purification. The {detailed structure of the 
purified state} is explained in the
following sections. Applications in proving steerability are also
discussed. Among these, a generalisation of the fundamental lemma of
the so-called `all-versus-nothing proof of
steerability'~\cite{chen_all-versus-nothing_2013} is proved in an
elementary way.

\subsection*{Effects and steered states} 
Consider the case where Alice and Bob share a state $\rho$ of a
bipartite quantum system over $\H_A \otimes \H_B$ with dimension $d_A
\times d_B$. Suppose Alice makes a positive operator valued measurement (POVM) described by effects $\{E_i\}_{i=1}^n$ on her system A, where the effects $E_i$
are positive semi-definite operators and $\sum_{i=1}^{n}E_i =
\II_A$.
Imagine that Alice can make all possible measurements. Then
the set of all possible effects is $\M_A = \{E
| 0 \le E \le \II_A\}$. Regardless of the POVM Alice makes, whenever she gets an effect $E\in \M_A$,  
Bob's system $B$ is \emph{steered} to a
conditional state $E'= \Tr_A [\rho (E \otimes \II_B)]$. The 
{set of} possible
conditional states Alice can steer Bob's system to is $\M_A'= \{E'| E \in
\M_A\}$~\cite{nguyen2016non,nguyen2016necessary}. A conditional state
$E'$ can also be characterised by {the}~\emph{steered state}
$\sigma = E'/\Tr(E')$ and {its}~\emph{steering probability}
$p=\Tr(E')$~\cite{jevtic_quantum_2014,nguyen2016non}. 

The set of steered states have been studied recently and a deep connection to the nonlocal properties of the shared bipartite state has been revealed~\cite{jevtic_quantum_2014,jevtic2015einstein,nguyen2016non,nguyen2016necessary}. In this work, we consider a special subset of steered states, the \emph{pure steered states}. For a general bipartite state $\rho$, there may exist no pure steered states. However when they do exist, they often carry interesting information about the shared state $\rho$.

\subsection*{Pure steered states and non-degenerate projective effects} 
To steer Bob's system to a specific state, we may imagine that Alice
tries to design measurements that are efficient in the sense that the
steering probability is maximal. We first note that to steer Bob's
system to a pure state, non-degenerate projective measurements are as efficient as general POVMs. 

Recall that a POVM $\{E_i\}_{i=1}^n$ is called a projection valued measurement (PVM), or \emph{projective measurement}, if the effects $E_i$ are mutual orthogonal projections. A projective measurement is \emph{non-degenerate} if all the orthogonal effects are rank-$1$. Effects of a non-degenerate projective measurement, which are of always the form $\ketbra{\alpha}{\alpha}$ for some vector $\ket{\alpha}$, are also called non-degenerate projective effects. 
\begin{lemma}
If Alice can steer Bob's system to a pure state, she can always gain the maximal steering probability for that steered state with a non-degenerate projective measurement.
\label{lem:pure_efficiency}
\end{lemma} 
\begin{proof}
Suppose Alice can steer Bob's system to $\ketbra{\beta}{\beta}$ with maximal probability $p$, then there exists an effect $E \in \M_A$ such that $E' = p \ketbra{\beta}{\beta}$. Let us consider the spectral decomposition of $E$, $E= \sum_{i=1}^{n} \lambda_i \ketbra{\alpha_i}{\alpha_i}$, where we kept only $\lambda_i > 0$. Accordingly, $E'=\sum_{i=1}^{n} \lambda_i (\ketbra{\alpha_i}{\alpha_i})' = p \ketbra{\beta}{\beta}$. But since $\ketbra{\beta}{\beta}$ is a pure state, and pure states are extremal~\cite{heinosaari_mathematical_2011}, this is only possible if for all $i$, either $(\ketbra{\alpha_i}{\alpha_i})'=0$ or $(\ketbra{\alpha_i}{\alpha_i})' \propto  \ketbra{\beta}{\beta}$. Suppose for $1 \le i \le m$, $(\ketbra{\alpha_i}{\alpha_i})'  \propto  \ketbra{\beta}{\beta} $, and for $m \le i \le n$, $(\ketbra{\alpha_i}{\alpha_i})' = 0$, then $\sum_{i=1}^{m} \lambda_i (\ketbra{\alpha_i}{\alpha_i})'= p \ketbra{\beta}{\beta}$. This in fact implies that $\lambda_i=1$ for $1\le i \le m$, otherwise we can always increase $\lambda_i$ to increase $p$. Thus $\sum_{i=1}^{m} (\ketbra{\alpha_i}{\alpha_i})'= p \ketbra{\beta}{\beta}$. This means that for a non-degenerate projective measurement that contains all $\{\ketbra{\alpha_i}{\alpha_i}\}_{i=1}^{m}$, Alice is able to steer Bob's system to $\ketbra{\beta}{\beta}$ with probability $p$. Note that all effects $\{\ketbra{\alpha_i}{\alpha_i}\}_{i=1}^{m}$ are mapped to the same steered state $\ketbra{\beta}{\beta}$, and the probability is accumulated to $p$.
\end{proof} 

The above lemma implies that Alice can always obtain the maximal steering probability of a pure steered state by accumulating multiple non-degenerate projective effects.

\subsection*{Behaviour of the purified state in steering} 
Further insight can be gained by considering the purification of the joint state. Let $C$ be an ancillary system attached to $AB$ such that the whole system is in a pure state $\ket{\Psi}$ and $\rho=\Tr_C(\ketbra{\Psi}{\Psi})$. Since $\ket{\Psi}$ is pure, after Alice gets an outcome in a non-degenerate projective measurement on $A$, $BC$ is also in a pure state. If further $B$ is in a pure state, so must be $C$. This observation leads directly to the following lemma.

\begin{lemma}
If with a non-degenerate projective effect $\ketbra{\alpha}{\alpha}$, Alice steers Bob's system to a pure state $\ketbra{\beta}{\beta}$ with probability $p$, then $C$ is also collapsed to a pure state $\ketbra{\gamma}{\gamma}$. Accordingly, the purified state $\ket{\Psi}$ can be written as
\begin{equation}
\ket{\Psi}= c \ket{\alpha,\beta,\gamma} + \ket{\tilde{\Phi}},
\label{eq:one_pure}
\end{equation}
where $\abs{c}^2= p$ and the partial projection of $\ket{\tilde{\Phi}}$ on $\ket{\alpha}$ vanishes, i.e., $\braket{\alpha}{\tilde{\Phi}}=0$. The tilde indicates that $\ket{\tilde{\Phi}}$ is not normalised.
\label{lem:one_pure}
\end{lemma}
\begin{proof}
Suppose Alice makes a non-degenerate projective measurement and gets $\ketbra{\alpha}{\alpha}$ with probability $p$, then the purified state can always be written as
\begin{equation}
\ket{\Psi}= c \ket{\alpha} \ket{\delta} + \ket{\tilde{\Phi}},
\end{equation}
where $\abs{c}^2= p$ and  $\braket{\alpha}{\tilde{\Phi}}=0$. Here $\ket{\delta}$ is a pure state of $BC$. However if $B$ is in pure state $\ket{\beta}$, we must have $\ket{\delta}=\ket{\beta}\ket{\gamma}$ and~\eqref{eq:one_pure} follows. 
\end{proof}

Using Lemma~\ref{lem:pure_efficiency}, one can easily extend Lemma~\ref{lem:one_pure} to the case where the pure conditional state $p \ketbra{\beta}{\beta}$ is obtained by accumulating several non-degenerate projective effects, $p \ketbra{\beta}{\beta} = \sum_{i=1}^{m} (\ketbra{\alpha_i}{\alpha_i})'$. In this case, the purified state can be written as 
\begin{equation}
\ket{\Psi}= \sum_{i=1}^{m} c_i \ket{\alpha_i,\beta,\gamma_i} + \ket{\tilde{\Phi}},
\end{equation}
where $\sum_{i=1}^{m} \abs{c_i}^2 = p$ and $\braket{\alpha_i}{\tilde{\Phi}}=0$ for all $i$. 

More interesting is the extension to the case of several different pure steered states. In this case, we need to consider the relationship between them. The following lemma is an obvious extension when several pure steered states are obtained in a single non-degenerate projective measurement.

\begin{lemma}
If with a non-degenerate projective measurement $\{\ketbra{\alpha_i}{\alpha_i}\}_{i=1}^{d_A}$, Alice can steer Bob's system to (not necessarily orthogonal) pure states $\{\ketbra{\beta_i}{\beta_i}\}_{i=1}^{d_A}$ with probabilities $\{p_i\}_{i=1}^{d_A}$, then $C$ is also collapsed to (not necessarily orthogonal) pure states $\{\ketbra{\gamma_i}{\gamma_i}\}_{i=1}^{d_A}$. Accordingly, the purified state $\ket{\Psi}$ can be written as
\begin{equation}
\ket{\Psi}= \sum_{i=1}^{d_A} c_i \ket{\alpha_i,\beta_i,\gamma_i}.
\label{eq:projective_connected_purified}
\end{equation}
with $\abs{c_i}^2=p_i$. 
\label{lem:projective_connected}
\end{lemma}
We ignore the proof, which is an easy extension of that of Lemma~\ref{lem:one_pure}.

\begin{example}
Consider the case where Alice and Bob share two qubits in state
\begin{align}
\rho&= \eta \ketbra{0,\beta_1}{0,\beta_1} + (1-\eta) \ketbra{1,\beta_2}{0,\beta_2} + \nonumber \\
& \quad \sqrt{\eta(1-\eta)} (z\ketbra{0,\beta_1}{1,\beta_2} + z^{\ast} \ketbra{1,\beta_2}{0,\beta_1}),
\label{eq:ex_2qbit}
\end{align}
where $0 \le \eta \le 1$, $\abs{z} \le 1$ and $\ket{\beta_1}$ and $\ket{\beta_2}$ are two arbitrary states. Note that for $z=0$ the state $\rho$ is separable and for $\eta= 1$ or $\eta=0$ it is pure. The state can be purified by attaching a third qubit. The purified state is
\begin{equation}
\ket{\Psi}= \sqrt{\eta} \ket{0,\beta_1,\gamma_1} + \sqrt{1-\eta} \ket{1,\beta_2,\gamma_2}, 
\label{eq:ex_2qbit_purified}
\end{equation}
with $\braket{\gamma_2}{\gamma_1}=z$. This is of the form~\eqref{eq:projective_connected_purified} and it is obvious that by measuring $\sigma_z$, Alice can steer Bob's system to pure states $\ket{\beta_1}$ or $\ket{\beta_2}$. 
\label{ex:2qubit}
\end{example}

The form~\eqref{eq:projective_connected_purified} of the purified
state turns out to be very informative. Later, using {Lemma~\ref{lem:projective_connected}}, we
will generalize the fundamental lemma of the all-versus-nothing proof
of steerability, originally designed only for two
qubits~\cite{chen_all-versus-nothing_2013}, to systems of arbitrary
dimension (Theorem~\ref{pros:hidden-state-model}). Before doing so, we
consider some other natural questions regarding the general properties
of the pure steered states.

The assumption in Lemma~\ref{lem:projective_connected} requires that the steered states are obtained in a single non-degenerate projective measurement. We ask the question if looking only at the steered states, is it possible to establish that several pure steered states are generated from a single non-degenerate projective measurement. We show that this is possible, albeit with rather restrictive assumptions.

\begin{lemma}
Consider a set of pure steered states $\{\ketbra{\beta_i}{\beta_i}\}_{i=1}^{n}$, each obtained from a single non-degenerate projective effect. In addition, if the states are orthogonal and have steering probabilities $p_i$ summed up to $1$, $ \sum_{i=1}^np_i=1$, then they can be obtained from a non-degenerate projective measurement.
\label{lem:orthogonal_complete}
\end{lemma}
Note that the assumption of this lemma is rather restrictive; in fact, pure steered states that are obtained 
by accumulating several non-degenerate projective effects are excluded from consideration.
\begin{proof}
Suppose with non-degenerate projective effects $\{\ketbra{\alpha_i}{\alpha_i}\}_{i=1}^{n}$, Alice can steer Bob's system to orthogonal states $\{\ketbra{\beta_i}{\beta_i}\}_{i=1}^{n}$ with probabilities $\{p_i\}_{i=1}^{n}$, $\sum_{i=1}^n p_i=1$. We can assume that $p_i$ are non-zero, since outcomes with empty probability can always be added to any measurement. We are going to show that $\{\ketbra{\alpha_i}{\alpha_i}\}_{i=1}^{n}$ are orthogonal. This would imply that we can extend $\{\ketbra{\alpha_i}{\alpha_i}\}_{i=1}^{n}$ to form a non-degenerate projective measurement.

According to Lemma~\ref{lem:one_pure}, for each effect $\ketbra{\alpha_{i}}{\alpha_{i}}$  we can write 
\begin{equation}
\ket{\Psi} = c_{i} \ket{\alpha_{i},\beta_i,\gamma_{i}} + \ket{\tilde{\Phi}_{i}}, 
\label{eq:for_each_i}
\end{equation}
where $\braket{\alpha_{i}}{\tilde{\Phi}_{i}}=0$ with $\abs{c_i}^2=p_i$.  Note that in the Hilbert space of $ABC$, the set $\{\ket{\alpha_{i},\beta_i,\gamma_i}\}_{i=1}^{n}$ are orthogonal because of the orthogonality of $\{\ket{\beta_i}\}_{i=1}^n$. Moreover $\sum_{i=1}^{n} \abs{c_i}^2= 1$. It then follows that 
\begin{equation}
\ket{\Psi} = \sum_{i=1}^{n} c_{i} \ket{\alpha_{i},\beta_i,\gamma_{i}},
\end{equation}
{by exhaustive expansion of $\ket{\Psi}$ with respect to
orthonormal states $\{\ket{\alpha_i,\beta_i,\gamma_i}\}_{i=1}^n$.}
Comparing with~\eqref{eq:for_each_i}, we get 
\begin{equation}
\ket{\tilde{\Phi}_i} = \sum_{j=1,j \ne i}^{n} c_{j} \ket{\alpha_{j},\beta_j,\gamma_{j}}.
\end{equation}
Since $\braket{\alpha_{i}}{\tilde{\Phi}_i}=0$, we have
\begin{equation}
0= \sum_{j=1,j \ne i}^{n} c_j \braket{\alpha_i}{\alpha_j} \ket{\beta_j,\gamma_j}.
\end{equation}
But since $\{\ket{\beta_j,\gamma_j} \}_{j=1,j \ne i}^{n}$ are orthogonal and $c_j \ne 0$, we must have $\braket{\alpha_i}{\alpha_j}=0$ for $j \ne i$, i.e., $\{\ketbra{\alpha_i}{\alpha_i}\}_{i=1}^{n}$ are orthogonal.
\end{proof}

\begin{example}
One can consider the state~\eqref{eq:ex_2qbit} over two qubits, now with $\ket{\beta_1}=\ket{0}$ and $\ket{\beta_2}=\ket{1}$. The two conditional states $(\ketbra{0}{0})'= \eta\ketbra{0}{0}$ and $(\ketbra{1}{1})'=(1-\eta) \ketbra{1}{1}$ are each obtained from a single non-degenerate projective effect. Moreover, they are orthogonal and have probabilities summed up to $1$. Therefore, they must be obtainable from a single non-degenerate projective measurement on $A$, which is $\sigma_z$ in this case. It is also clear from this example that the assumption in Lemma~\ref{lem:orthogonal_complete}, which requires $\braket{\beta_1}{\beta_2}=0$, is more restrictive than Lemma~\ref{lem:projective_connected}.
\end{example}

\subsection*{The emergence of subspaces of pure steered states}
Suppose Alice and Bob share two qubits which are in a pure nonseparable state. 
When Alice performs a local non-degenerate projective measurement on her system, Bob's system is steered to a pure state. 
With all possible non-degenerate projective effects, Alice can steer Bob's system to all possible pure states with different probabilities. We see that Alice's pure steered states cover the whole Hilbert space of Bob's system. 

We will see that in general Alice's steered states may cover only a subspace of Bob's Hilbert space.
This notion of subspace of pure steered states also appears when the shared bipartite system is in a suitable mixed state. In fact, these mixed states can be viewed as `partial pure states'--states that look like pure ones in some restricted subspace. The following theorem allows us to characterise these particular mixed states via their behaviour in EPR experiments.  

\begin{theorem}
\label{pros:pure-state}
Suppose the set of steered states on Bob's side contains a subset of pure steered states $\{\ketbra{\beta_i}{\beta_i}\}_{i=1}^{n}$, each obtained from a single non-degenerate projective effect. In addition, suppose that these steered states are orthogonal and have steering probabilities $p_i$ summed up to $1$, $ \sum_{i=1}^np_i=1$. Now if Alice can steer Bob's system
to another pure state $\ket{\beta}$ by a non-degenerate projective effect, then she can also steer Bob's system to any state in $\Span(\{\ket{\beta_i} \vert \braket{\beta_i}{\beta} \ne 0\})$.
\end{theorem}
Note that here and from now on we also use vectors such as $\ket{\beta}$ on behalf of their projections  $\ketbra{\beta}{\beta}$ to indicate pure states.


In order to prove this theorem, we need the following simple lemma.

\begin{lemma}
\label{lem:separable-plane}
Let $\{\ket{\alpha_i}\}_{i=1}^{n}$ be a set of
$n \le d_A$ orthonormal states of system $A$, and $\{\ket{\beta_i}\}_{i=1}^{n}$ be arbitrary states of system $B$. If there {exist} $n$ non-zero
numbers $a_i$ such that the linear combination $\sum_{i=1}^{n} a_i
\ket{\alpha_i} \ket{\beta_i}$ is a product state, then
$\ket{\beta_1}=\ket{\beta_2}=\cdots=\ket{\beta_n}$.
\end{lemma}
\begin{proof}
Suppose $\sum_{i=1}^{n} a_i \ket{\alpha_i}
\ket{\beta_i}$ is a product state, i.e., $\sum_{i=1}^{n} a_i \ket{\alpha_i}
\ket{\beta_i}=\ket{\alpha} \ket{\beta}$. Then when
tracing over A, one gets
\begin{equation}
\ketbra{\beta}{\beta}=\sum_{i=1}^{n} a_i a_i^{\ast} \ketbra{\beta_i}{\beta_i},
\end{equation}
where we have used the fact that $\{\ket{\alpha_i}\}_{i=1}^{n}$ are
orthogonal. Since pure states are extremal states~\cite{heinosaari_mathematical_2011}, this only happens for non-zero $a_i$ if $\ket{\beta_1}=\ket{\beta_2}=\cdots=\ket{\beta_n}$.
\end{proof}

\begin{proof}[Proof of theorem~\ref{pros:pure-state}]
Let us attach a system $C$ to purify $AB$ to a state $\ket{\Psi}$. 
Because of lemma~\ref{lem:orthogonal_complete}, the purified state takes the form
\begin{equation}
\ket{\Psi}= \sum_{i=1}^{n} c_{i} \ket{\alpha_{i},\beta_{i},\gamma_{i}},
\label{eq:particular-form}
\end{equation}
where $\{\ket{\alpha_{i}}\}_{i=1}^{n}$ are orthogonal. Now if Alice can
steer $B$ to another state $\ket{\beta}$ 
by a non-degenerate projective effect, by
Lemma~\ref{lem:one_pure}, there exists state 
$\ket{\alpha}$ in $\hi_A$
and state $\ket{\gamma}$ in $\hi_C$ such that
$\braket{\alpha}{\Psi}=c \ket{\beta}\ket{\gamma}$ for some $c \ne
0$. Using (\ref{eq:particular-form}), one has 
\begin{equation}
\ket{\beta,\gamma}=\sum_{i=1}^{n} \frac{c_i \braket{\alpha}{\alpha_{i}}}{c}
\ket{\beta_{i},\gamma_{i}}.
\label{eq:consequence}
\end{equation} 
Since $\{\ket{\beta_i}\}_{i=1}^{n}$ are orthogonal by assumption, one finds
$\braket{\beta_{i}}{\beta}=\frac{c_{i}}{c}
\braket{\alpha}{\alpha_i} \braket{\gamma}{\gamma_i}$. For any
$i$ such that $\braket{\beta_{i}}{\beta} \ne 0$, one has
$\frac{c_{i}}{c} \braket{\alpha}{\alpha_{i}} \ne 0$. Thus by
Lemma~\ref{lem:separable-plane},~\eqref{eq:consequence} implies that for all $\ket{\beta_{i}}$ such that
$\braket{\beta_{i}}{\beta} \ne 0$,
$\ket{\gamma_{i}}=\ket{\gamma}$. 
Without loss of generality, we assume that  $\braket{\beta_i}{\beta} \ne 0$ for $1 \le i \le m$ and
$\braket{\beta_i}{\beta} = 0$ for $m <i \le n$. The purified
state $\ket{\Psi}$ can now be written as
\begin{align}
\ket{\Psi}&= \left( \sum_{i=1}^{m} c_i \ket{\alpha_{i},\beta_{i}} \right) \ket{\gamma} + \sum_{i=m+1}^{n} c_{i} \ket{\alpha_{i},\beta_{i},\gamma_{i}}.
\label{eq:k-decomposition}
\end{align}
Note that $\{\ket{\alpha_i}\}_{i=1}^{m}$ are orthogonal due to
Lemma~\ref{lem:orthogonal_complete}.
We claim that with this form, the bipartite
state allows Alice to steer $B$ to any state in $\Span
(\{\ket{\beta_i}\}_{i=1}^{m})$. To see that we just calculate the (unnormalised) state of Bob's system given that Alice gets state $\ket{\alpha}=\sum_{i=1}^{m} a_i \ket{\alpha_i}$ in a non-degenerate projective measurement,
\begin{align}
\braket{\alpha}{\Psi}= \left( \sum_{i=1}^{m} a_i^{\ast} c_i \ket{\beta_i} \right) \ket{\gamma},
\label{eq:k-decomposition-rotated}
\end{align}
where the separated ancillary state $\ket{\gamma}$ can be simply ignored. 
Since $c_i \ne 0$ and $a_i$ are arbitrary, the conditional states~\eqref{eq:k-decomposition-rotated} cover the whole $\Span
(\{\ket{\beta_i}\}_{i=1}^{m})$.
\end{proof}

\begin{example}
When applied to a two-qubit system, this theorem is particularly simple. For a two-qubit system, suppose by non-degenerate projective measurements, Alice can steer Bob's system to two orthogonal pure states with probabilities summed up to
$1$, and another additional pure state, then $AB$ must be in a pure state. This can be an operational way for Alice to prove to Bob that their shared state is a pure entangled one (instead of steering Bob's system to infinitely many pure states). 
\end{example}

\begin{example}
Consider the case where Alice and Bob share two qutrits in state
\begin{align}
\rho=& \eta \ketbra{\psi^{-}}{\psi^{-}} + (1-\eta) \ketbra{0,0}{0,0} +  \nonumber \\ & \sqrt{\eta (1-\eta)} (z \ketbra{\psi^{-}}{0,0}+ z^{\ast} \ketbra{0,0}{\psi^{-}}), 
\label{eq:qutrit-pair}
\end{align}
with $\ket{\psi^{-}}= \frac{1}{\sqrt{2}} (\ket{+1,-1} - \ket{-1,+1})$, $0 < \eta < 1$ and $\abs{z} \le 1$. The state can be purified by attaching a qubit. The purified state is 
\begin{equation}
\ket{\Psi} = \sqrt{\eta} \ket{\psi^{-}} \ket{\gamma_1} + \sqrt{1-\eta} \ket{0,0}\ket{\gamma_2},
\end{equation} 
with $\braket{\gamma_2}{\gamma_1}=z$, which is clearly of the form~\eqref{eq:k-decomposition}. It is then easy to see that Alice can steer Bob system to any linear combination of $\ket{-1}$ and $\ket{+1}$.
\label{ex:qutrit-pair}
\end{example}

\subsection*{Pure steered states and steerability} 
In their paper, Wiseman~\emph{et al.}~\cite{wiseman_steering_2007} discovered that only for certain states the EPR steering experiment is verifiable to be truly nonlocal, while for other states the steering experiment can actually be locally simulated. The former are regarded as \emph{steerable} states, and the latter are called \emph{unsteerable} states. 

To demonstrate steering, Alice prepares a bipartite quantum state
$\rho$ over $\H_A \otimes \H_B$. She sends part $B$ to Bob and
specifies a set of measurements $A$ (called \emph{measurement
assemblage}) she can make. Bob then asks her to make a measurement
$A_x$ from $A$, which consists of effects $\{A_{a|x}\}_a$.  As we
described above, it is expected that upon Alice making the
measurement, Bob's system is steered to unnormalised conditional
states $A'_x= \{\Tr_A [\rho(A_{a|x}\otimes \II)] \}_a$, which form an
ensemble $\{P(a|x),\rho_{a|x}\}_a$ with $P(a|x)=\Tr(A'_{a|x})$ and
$\rho_{a|x}=A'_{a|x}/P(a|x)$. When receiving the outcome from Alice,
Bob can perform state tomography to verify his corresponding
conditional state. Being able to perform any specified measurement
$x$, Alice intends to convince Bob that she can steer his system to
different ensembles in the steering assemblage $A'$ from {a} distance.

However, when Bob does not trust Alice, then this verification protocol may not be always convincing. Indeed, there may exist a strategy for Alice to cheat Bob. In the cheating strategy, instead of sending Bob halves of entangled systems, Alice sends him random states choosing from an ensemble of \emph{Local Hidden States (LHS)} $\{P(\xi),\sigma_\xi\}_{\xi}$. Upon receiving the request of measuring $x$, she simulates the measurement outcomes by randomly choosing an outcome $a$ with a designed probability function $P(a|x,\xi)$. Such a simulation gives exactly the expected result of state tomography by Bob if
\begin{align}\label{eq:hidden-model}
  A'_{a|x}=\sum_{\xi}P(a|x,\xi)P(\xi)\sigma_\xi,
\end{align}
in which case the state is \emph{unsteerable} (here always considered from Alice{'s} side) with measurement assemblage $A$. 
It is clear that the existence of the LHS ensemble, which allows for this cheating strategy, depends crucially on the steering assemblage $A'$. The steering assemblage $A'$, in turn, depends on both the shared state $\rho$ and the measurement assemblage $A$. 
In the following, if the measurement assemblage is not explicitly specified, we assume that it consists of all non-degenerate projective measurements. 

Following the above consideration, we call $\mathcal{A}=\{A_{a|x}\}_{a,x}$ the set of effects for the assemblage $A$, which is a subset of $\M_A$. Accordingly, $\mathcal{A}'=\{A_{a|x}'\}_{a,x}$ is a subset of $\M_A'$, called the set of conditional states on Bob's system for the assemblage $A$. The problem of determining the existence of a LHS ensemble is also greatly simplified if $\mathcal{A}'$ contains one or several (unnormalised) pure states. At the heart of this
simplification is the following lemma.

\begin{lemma}
\label{lem:pure-hidden-model}
If Alice can steer Bob's system to a pure state 
$\ket{\beta}\bra{\beta}$ 
with probability $p$, then the LHS ensemble, if
exists, must contain 
$\ket{\beta}\bra{\beta}$ with probability no less than $p$.
\end{lemma}
\begin{proof}
  By assumption, in equation~(\ref{eq:hidden-model}), $A_{a|x}'=p\ket{\beta}\bra{\beta}$ for some $a$ and $x$. But since pure states are
  {extremal} points of the set of states of $B$~\cite{heinosaari_mathematical_2011}, 
  the decomposition in terms of states of the LHS 
  ensemble~(\ref{eq:hidden-model}) is only possible if
  $\sigma_{\xi_0}=\ket{\beta}\bra{\beta}$ for some $\xi_0$ and $P(a|x,\xi)$ all vanish except for $\xi = \xi_0$. Moreover $P(a|x,\xi_0) \le 1$, thus $P(\xi_0) \ge p$.
\end{proof}

Despite its simplicity, Lemma~\ref{lem:pure-hidden-model} is in fact
very useful. As a simple corollary, if Alice can steer Bob to some
pure states (not necessarily orthogonal) with probabilities summed up
greater than $1$, then they cannot be contained in any LHS ensemble and the state of $AB$ is steerable from Alice's
side. 

%
\begin{example} 
For a pure nonseparable state, Alice can steer Bob's system to infinitely many pure states by using different non-degenerate projective measurements. The total probability is infinite, therefore nonseparable pure states are all steerable. 
\end{example}

\begin{example} 
One can take the state~\eqref{eq:qutrit-pair} in example~\ref{ex:qutrit-pair} with $0 < \eta < 1$. Although the state is mixed, the subspace of pure steered state is infinite. Steerability is therefore also guaranteed. Although we have obtained the results with fairly elementary arguments, it is interesting to note that the steerability of the state is not easily detected by steering inequalities. 
For example, consider the steering inequality proposed in \cite{he_entanglement_2011},
\begin{align}\label{eq:2} \abs{\langle S^{s_A}\otimes
S^{s_B}\rangle}^2> &\langle
((S^x)^2+(S^y)^2-\frac{7}{16})\otimes((S^x)^2+(S^y)^2)\rangle,
\end{align} 
where $\langle K_{AB}\rangle =\tr{K_{AB}\rho_{AB}}$, $S^{s_{A/B}}$ with $s_A,s_B=\pm$  are the spin raising and lowering operators, and $S^i$ with $i=x,y,z$ are the usual self-adjoint
spin operators. In the inequality, $s_A$ and $s_B$ can be chosen independently to optimize
the possible violation. When the inequality is true, it witnesses the
steerability of the state.  For simplicity, we consider a simple case of~\eqref{eq:qutrit-pair} where $z=0$,
\begin{align}\label{eq:3} 
\rho=& \eta \ketbra{\psi^-}{\psi^-} +(1-\eta)\ketbra{00}{00}.
\end{align} 
For this state, the left hand-side of equation~(\ref{eq:2}) gives $0$ and the right hand-side gives
$\frac{1}{16}(50-41 \eta)$. Thus the inequality is not violated. Note also that this example can also be easily extended to systems of arbitrary dimension. 

\end{example}

%
%

Following this line of argument, when the set of conditional states $\mathcal{A}'$ contains a set of pure states with steering
probabilities summed up to $1$, the LHS ensemble, if exists, can only
be these pure states with their corresponding steering
probabilities. Determining the steerability of the state is therefore
reduced to checking if the LHS ensemble can explain all measurements in the assemblage $A$ in the sense of equation~(\ref{eq:hidden-model}). The following theorem utilises the idea.

\begin{theorem}
\label{pros:hidden-state-model}
If by a non-degenerate projective measurement Alice can steer Bob's system $B$ to a set of independent pure states, then the joint state of $AB$ is either
separable or steerable.
\end{theorem} 

\begin{lemma}
\label{lem:independent}
Let $\{\ket{\beta_i}\}_{i=1}^{n}$ be independent states of
some quantum system, then $\{E_{ij}=\ketbra{\beta_i}{\beta_j}\}_{i,j=1}^{n}$ are independent
operators.
\end{lemma}
\begin{proof}
We can suppose that $\{\ket{\beta_i}\}_{i=1}^{n}$ is a basis of the
system; otherwise we can extend it to form a basis. Let
$\{\ket{\beta'_j}\}_{j=1}^{n}$ be the dual basis, defined by
$\braket{\beta'_i}{\beta_j}=\delta_{ij}$. Consider a linear
combination $Q=\sum_{i,j=1}^{n} \lambda_{ij}E_{ij}$, we observe that
$\bra{\beta'_i}Q\ket{\beta'_j}=\lambda_{ij}$. It follows that $Q=0$ if and only if
$\lambda_{ij}=0$ for all $i$ and $j$, or $\{E_{ij}\}_{i,j=1}^{n}$ are
independent.
\end{proof}

\begin{proof}[Proof of theorem \ref{pros:hidden-state-model}]
Suppose the joint state is unsteerable, we will show that it is
separable. Suppose Alice performs the non-degenerate projective measurement
$A_x=\{\ketbra{\alpha_{a|x}}{\alpha_{a|x}}\}_{a=1}^{d_A}$
locally. After the measurement, Alice's state is described by  
an ensemble $\{P(a|x),\ketbra{\alpha_{a|x}}{\alpha_{a|x}}\}_{a=1}^{d_A}$ and
correspondingly, Bob's system is steered to an ensemble of
not necessarily orthogonal pure states
$\{P(a|x),\ketbra{\beta_{a|x}}{\beta_{a|x}}\}_{a=1}^{d_A}$. Note that $\sum_{a} P(a|x)=1$.
It follows from our above discussion that the LHS ensemble must be $\{P(a|x),\ketbra{\beta_{a|x}}{\beta_{a|x}}\}_{a=1}^{d_A}$. Obviously, we have a
very important information.

Let us attach a system $C$ to purify the state $\rho$ of $AB$ to a pure state
$\ket{\Psi}$ of $ABC$. From Lemma~\ref{lem:projective_connected}, we know that the purified state must be of the form
\begin{equation}
\ket{\Psi}= \sum_{a=1}^{d_A} c_{a | x} \ket{\alpha_{a|x},\beta_{a|x},\gamma_{a|x}},
\label{eq:purified_form}
\end{equation}
with $\abs{c_{a|x}}^2=P(a|x)$. Then let Alice make another measurement $A_{x'}$ in the assemblage $A$ (here being all possible non-degenerate projective measurements), which is a collection of new rank-$1$ projections $\{A_{a'|x'}=\ketbra{\alpha_{a'|x'}}{\alpha_{a'|x'}}\}_{a'=1}^{d_A}$. The new rank-$1$ projections are
related to $\{\ketbra{\alpha_{a|x}}{\alpha_{a|x}}\}_{a=1}^{d_A}$ by a $d_A \times d_A $
unitary matrix $U$,
\begin{equation}
  \ket{\alpha_{a'|x'}} = \sum_{a=1}^{d_A}  \ket{\alpha_{a|x}} U_{aa'}.
\end{equation} 
Now the joint state can be expressed using  $\{\ket{\alpha_{a'|x'}}\}_{a'=1}^{d_A}$ as
\begin{equation}
\ket{\Psi}=\sum_{a=1}^{d_A}\sum_{a'=1}^{d_A} c_{a|x}U_{aa'}^{\ast} \ket{\alpha_{a'|x'},\beta_{a|x},\gamma_{a|x}}.
\end{equation}
If in the measurement $x'$, Alice gets the outcome $a'$, $BC$ is then steered to 
\begin{equation}
\sum_{a=1}^{d_A} c_{a|x} U^{\ast}_{aa'} \ket{\beta_{a|x},\gamma_{a|x}}.
\end{equation}
After tracing out the ancillary system $C$, one gets the corresponding conditional state of $B$,
\begin{equation}
A_{a'|x'}' = \sum_{b,d=1}^{d_A}  c_{b|x}c_{d|x}^* U_{ba'}^*U_{da'} \braket{\gamma_{d|x}}{\gamma_{b|x}} \ketbra{\beta_{b|x}}{\beta_{d|x}}. 
\label{eq:decomposition}
\end{equation}
Note that $\{\ket{\beta_{a|x}}\}_{a=1}^{d_A}$ are independent, thus so are
$\{\ketbra{\beta_{b|x}}{\beta_{d|x}}\}_{b,d=1}^{d_A}$ by
Lemma~\ref{lem:independent}. The coefficients in the decomposition~(\ref{eq:decomposition}) are therefore unique.

Since the $\{P(a|x),\ketbra{\beta_{a|x}}{\beta_{a|x}}\}_{a=1}^{d_A}$ is the LHS ensemble, the conditional states $A_{a'|x'}'$ must be expanded by rank-$1$ projections
$\{\ketbra{\beta_{a|x}}{\beta_{a|x}}\}_{b=1}^{d_A}$ for all unitary matrices $U$
and for all outcomes $a'$.  This means those operators
$\ketbra{\beta_{b|x}}{\beta_{d|x}}$ with $b \ne d$ must be absent from the
decomposition~(\ref{eq:decomposition}), i.e.,
\begin{equation}
 c_{b|x}c^{\ast}_{d|x}U_{ba'}^*U_{da'} \braket{\gamma_{d|x}}{\gamma_{b|x}} = 0. 
 \label{eq:off_diagonal_vanish}
\end{equation}
Now for a particular pair $b \ne d$, one can choose some $a'$ and a
unitary matrix $U$ such that $U_{ba'}^{\ast} U_{da'} \ne
0$. (For example, one can choose $a'=b$ and the matrix $U$ such
that $U_{bb}=U_{bd}=1/\sqrt{2}=U_{dd}=-U_{db}$ and {identity blocks} everywhere else.) This means that~\eqref{eq:off_diagonal_vanish} is satisfied only when $c_{b|x} c_{d|x}^{\ast}
\braket{\gamma_{d|x}}{\gamma_{b|x}}=0$ for all $b \ne d$. This makes the
state of $AB$, $\rho=\Tr_{C}
(\ketbra{\Psi}{\Psi})$, manifestly separable: $\rho=\sum_{b,d=1}^{d_A} c_{b|x}c_{d|x}^{\ast}
\braket{\gamma_{d|x}}{\gamma_{b|x}} \ketbra{\alpha_{b|x},\beta_{b|x}}{\alpha_{d|x},
\beta_{d|x}} = \sum_{b=1}^{d_A} \abs{c_{b|x}}^2 \ketbra{\alpha_{b|x},\beta_{b|x}}{\alpha_{b|x},
\beta_{b|x}}$.
\end{proof}
In fact, from this proof, a more detailed statement can be made.
\begin{corollary*}
Suppose by making a non-degenerate projective measurement consisting of orthonormal states $\{\ket{\alpha_i}\}_{i=1}^{d_A}$, Alice can steer Bob system to independent pure states $\{\ket{\beta_i}\}_{i=1}^{d_A}$, then either:
\begin{itemize}
\item[(i)] the joint state is steerable (thus nonseparable), in which case $\bra{\alpha_i,\beta_i} \rho \ket{\alpha_j,\beta_j}$ must not vanish for some $j \ne i$;
\item[(ii)] the joint state $\rho$ is separable of the form $\rho = \sum_{i=1}^{d_A} p_i \ketbra{\alpha_i,\beta_i}{\alpha_i,\beta_i}$.
\end{itemize}
\end{corollary*}

\begin{proof}
Again applying Lemma~\ref{lem:projective_connected}, the assumption of the corollary implies that the purified state is of the form~\eqref{eq:purified_form}, 
\begin{equation}
\ket{\Psi}= \sum_{i=1}^{d_A} c_i \ket{\alpha_i,\beta_i,\gamma_i}.
\end{equation}
By tracing over the ancillary system, we find
\begin{equation}
\rho= \sum_{i=1,j=1}^{d_A} c_i c_j^{\ast} \braket{\gamma_j}{\gamma_i} \ketbra{\alpha_i,\beta_i}{\alpha_j,\beta_j}.
\end{equation}
Note that $\bra{\alpha_i,\beta_i} \rho \ket{\alpha_j,\beta_j} = c_i c_j^{\ast} \braket{\gamma_j}{\gamma_i}$. Then it is clear that if $\bra{\alpha_i,\beta_i} \rho \ket{\alpha_j,\beta_j} \ne 0$ for some $j \ne i$, the state is steerable (thus nonseparable) according to the above proof; else if $\bra{\alpha_i,\beta_i} \rho \ket{\alpha_j,\beta_j} = 0$ for all $j \ne i$, $\rho$ reduces to the separable form stated in (ii) with $p_i= \abs{c_i}^2$.
\end{proof}

When applied to a two-qubit system, the statement is particularly
simple. For a two-qubit system, if by a non-degenerate projective measurement Alice
can steer $B$ to two pure states, then the joint state of $AB$ is either
separable or steerable. This result has been recently obtained by Chen~\textit{et al.}~\cite{chen_all-versus-nothing_2013}. Experimental applications of this special case were also discussed~\cite{sun_experimental_2014}. 
Here we report a more general result with a somewhat more systematic proof.

\begin{example}
Take again the two-qubit state~\eqref{eq:ex_2qbit} in example~\ref{ex:2qubit}, which satisfies the assumption of the above corollary. Note that $\bra{0,\beta_1} \rho  \ket{1,\beta_2} = z \sqrt{\eta (1-\eta)}$. We then know that it is steerable (or equivalently, nonseparable) if and only if $\abs{z}^2 \eta(1-\eta) > 0$. 
\end{example}

\subsection*{Conclusion}
We have shown that the pure steered states in the EPR experiments
carry interesting information about the shared state, in particular
its nonlocal {properties}. Thanks to the purification technique, results
are obtained in an {elementary} and systematic way. Although our work only
concentrates on the pure steered states, we speculate that
analysing the behaviour of the purified system in the EPR experiments
might be an interesting approach to quantum nonlocality. Applications
of the special case of Theorem~\ref{pros:hidden-state-model} for
two-qubit systems have been
discussed~\cite{chen_all-versus-nothing_2013,sun_experimental_2014}. As
experimental research is moving toward higher dimensional systems, we
hope that our generalisation will be useful for the future
experiments.
\begin{acknowledgments}
We thank David Gross, Eric Lutz, Roope Uola and Huangjun Zhu for useful discussions. The authors also acknowledge the Referee of PRA for the useful comments and particularly for correcting our terminology. 
\end{acknowledgments}

\bibliography{pure-state-steering}

\begin{thebibliography}{29}%
\makeatletter
\providecommand \@ifxundefined [1]{%
 \@ifx{#1\undefined}
}%
\providecommand \@ifnum [1]{%
 \ifnum #1\expandafter \@firstoftwo
 \else \expandafter \@secondoftwo
 \fi
}%
\providecommand \@ifx [1]{%
 \ifx #1\expandafter \@firstoftwo
 \else \expandafter \@secondoftwo
 \fi
}%
\providecommand \natexlab [1]{#1}%
\providecommand \enquote  [1]{``#1''}%
\providecommand \bibnamefont  [1]{#1}%
\providecommand \bibfnamefont [1]{#1}%
\providecommand \citenamefont [1]{#1}%
\providecommand \href@noop [0]{\@secondoftwo}%
\providecommand \href [0]{\begingroup \@sanitize@url \@href}%
\providecommand \@href[1]{\@@startlink{#1}\@@href}%
\providecommand \@@href[1]{\endgroup#1\@@endlink}%
\providecommand \@sanitize@url [0]{\catcode `\\12\catcode `\$12\catcode
  `\&12\catcode `\#12\catcode `\^12\catcode `\_12\catcode `\%12\relax}%
\providecommand \@@startlink[1]{}%
\providecommand \@@endlink[0]{}%
\providecommand \url  [0]{\begingroup\@sanitize@url \@url }%
\providecommand \@url [1]{\endgroup\@href {#1}{\urlprefix }}%
\providecommand \urlprefix  [0]{URL }%
\providecommand \Eprint [0]{\href }%
\providecommand \doibase [0]{http://dx.doi.org/}%
\providecommand \selectlanguage [0]{\@gobble}%
\providecommand \bibinfo  [0]{\@secondoftwo}%
\providecommand \bibfield  [0]{\@secondoftwo}%
\providecommand \translation [1]{[#1]}%
\providecommand \BibitemOpen [0]{}%
\providecommand \bibitemStop [0]{}%
\providecommand \bibitemNoStop [0]{.\EOS\space}%
\providecommand \EOS [0]{\spacefactor3000\relax}%
\providecommand \BibitemShut  [1]{\csname bibitem#1\endcsname}%
\let\auto@bib@innerbib\@empty
\bibitem [{\citenamefont {Einstein}\ \emph {et~al.}(1935)\citenamefont
  {Einstein}, \citenamefont {Podolsky},\ and\ \citenamefont
  {Rosen}}]{einstein_can_1935}%
  \BibitemOpen
  \bibfield  {author} {\bibinfo {author} {\bibfnamefont {A.}~\bibnamefont
  {Einstein}}, \bibinfo {author} {\bibfnamefont {B.}~\bibnamefont {Podolsky}},
  \ and\ \bibinfo {author} {\bibfnamefont {N.}~\bibnamefont {Rosen}},\
  }\bibfield  {title} {\enquote {\bibinfo {title} {Can quantum-mechanical
  description of physical reality be considered complete?}}\ }\href {\doibase
  10.1103/PhysRev.47.777} {\bibfield  {journal} {\bibinfo  {journal} {Phys.
  Rev.}\ }\textbf {\bibinfo {volume} {47}},\ \bibinfo {pages} {777--780}
  (\bibinfo {year} {1935})}\BibitemShut {NoStop}%
\bibitem [{\citenamefont
  {Schr{\"o}dinger}(1935)}]{schrodinger_discussion_1935}%
  \BibitemOpen
  \bibfield  {author} {\bibinfo {author} {\bibfnamefont {E.}~\bibnamefont
  {Schr{\"o}dinger}},\ }\bibfield  {title} {\enquote {\bibinfo {title}
  {Discussion of probability relations between separated systems},}\ }\href
  {\doibase 10.1017/S0305004100013554} {\bibfield  {journal} {\bibinfo
  {journal} {Proc. Cambridge Philos. Soc.}\ }\textbf {\bibinfo {volume} {31}},\
  \bibinfo {pages} {555--563} (\bibinfo {year} {1935})}\BibitemShut {NoStop}%
\bibitem [{\citenamefont {Wiseman}\ \emph {et~al.}(2007)\citenamefont
  {Wiseman}, \citenamefont {Jones},\ and\ \citenamefont
  {Doherty}}]{wiseman_steering_2007}%
  \BibitemOpen
  \bibfield  {author} {\bibinfo {author} {\bibfnamefont {H.~M.}\ \bibnamefont
  {Wiseman}}, \bibinfo {author} {\bibfnamefont {S.~J.}\ \bibnamefont {Jones}},
  \ and\ \bibinfo {author} {\bibfnamefont {A.~C.}\ \bibnamefont {Doherty}},\
  }\bibfield  {title} {\enquote {\bibinfo {title} {Steering, entanglement,
  nonlocality, and the {Einstein}-{Podolsky}-{Rosen} paradox},}\ }\href
  {\doibase 10.1103/PhysRevLett.98.140402} {\bibfield  {journal} {\bibinfo
  {journal} {Phys. Rev. Lett.}\ }\textbf {\bibinfo {volume} {98}},\ \bibinfo
  {pages} {140402} (\bibinfo {year} {2007})}\BibitemShut {NoStop}%
\bibitem [{\citenamefont {Jones}\ \emph {et~al.}(2007)\citenamefont {Jones},
  \citenamefont {Wiseman},\ and\ \citenamefont
  {Doherty}}]{jones_entanglement_2007}%
  \BibitemOpen
  \bibfield  {author} {\bibinfo {author} {\bibfnamefont {S.~J.}\ \bibnamefont
  {Jones}}, \bibinfo {author} {\bibfnamefont {H.~M.}\ \bibnamefont {Wiseman}},
  \ and\ \bibinfo {author} {\bibfnamefont {A.~C.}\ \bibnamefont {Doherty}},\
  }\bibfield  {title} {\enquote {\bibinfo {title} {Entanglement,
  {Einstein}-{Podolsky}-{Rosen} correlations, {Bell} nonlocality, and
  steering},}\ }\href {\doibase 10.1103/PhysRevA.76.052116} {\bibfield
  {journal} {\bibinfo  {journal} {Phys. Rev. A}\ }\textbf {\bibinfo {volume}
  {76}},\ \bibinfo {pages} {052116} (\bibinfo {year} {2007})}\BibitemShut
  {NoStop}%
\bibitem [{\citenamefont {Werner}(1989)}]{werner_quantum_1989}%
  \BibitemOpen
  \bibfield  {author} {\bibinfo {author} {\bibfnamefont {R.~F.}\ \bibnamefont
  {Werner}},\ }\bibfield  {title} {\enquote {\bibinfo {title} {Quantum states
  with {Einstein}-{Podolsky}-{Rosen} correlations admitting a hidden-variable
  model},}\ }\href {\doibase 10.1103/PhysRevA.40.4277} {\bibfield  {journal}
  {\bibinfo  {journal} {Phys. Rev. A}\ }\textbf {\bibinfo {volume} {40}},\
  \bibinfo {pages} {4277--4281} (\bibinfo {year} {1989})}\BibitemShut {NoStop}%
\bibitem [{\citenamefont {Bell}(1964)}]{bell_einstein-podolsky-rosen_1964}%
  \BibitemOpen
  \bibfield  {author} {\bibinfo {author} {\bibfnamefont {J.}~\bibnamefont
  {Bell}},\ }\bibfield  {title} {\enquote {\bibinfo {title} {On the
  {Einstein}-{Podolsky}-{Rosen} paradox},}\ }\href@noop {} {\bibfield
  {journal} {\bibinfo  {journal} {Physics}\ }\textbf {\bibinfo {volume} {1}},\
  \bibinfo {pages} {195--200} (\bibinfo {year} {1964})}\BibitemShut {NoStop}%
\bibitem [{\citenamefont {Cavalcanti}\ \emph {et~al.}(2009)\citenamefont
  {Cavalcanti}, \citenamefont {Jones}, \citenamefont {Wiseman},\ and\
  \citenamefont {Reid}}]{cavalcanti_experimental_2009}%
  \BibitemOpen
  \bibfield  {author} {\bibinfo {author} {\bibfnamefont {E.~G.}\ \bibnamefont
  {Cavalcanti}}, \bibinfo {author} {\bibfnamefont {S.~J.}\ \bibnamefont
  {Jones}}, \bibinfo {author} {\bibfnamefont {H.~M.}\ \bibnamefont {Wiseman}},
  \ and\ \bibinfo {author} {\bibfnamefont {M.~D.}\ \bibnamefont {Reid}},\
  }\bibfield  {title} {\enquote {\bibinfo {title} {Experimental criteria for
  steering and the {Einstein}-{Podolsky}-{Rosen} paradox},}\ }\href {\doibase
  10.1103/PhysRevA.80.032112} {\bibfield  {journal} {\bibinfo  {journal} {Phys.
  Rev. A}\ }\textbf {\bibinfo {volume} {80}},\ \bibinfo {pages} {032112}
  (\bibinfo {year} {2009})}\BibitemShut {NoStop}%
\bibitem [{\citenamefont {Cavalcanti}\ \emph {et~al.}(2011)\citenamefont
  {Cavalcanti}, \citenamefont {He}, \citenamefont {Reid},\ and\ \citenamefont
  {Wiseman}}]{cavalcanti_unified_2011}%
  \BibitemOpen
  \bibfield  {author} {\bibinfo {author} {\bibfnamefont {E.~G.}\ \bibnamefont
  {Cavalcanti}}, \bibinfo {author} {\bibfnamefont {Q.~Y.}\ \bibnamefont {He}},
  \bibinfo {author} {\bibfnamefont {M.~D.}\ \bibnamefont {Reid}}, \ and\
  \bibinfo {author} {\bibfnamefont {H.~M.}\ \bibnamefont {Wiseman}},\
  }\bibfield  {title} {\enquote {\bibinfo {title} {Unified criteria for
  multipartite quantum nonlocality},}\ }\href {\doibase
  10.1103/PhysRevA.84.032115} {\bibfield  {journal} {\bibinfo  {journal} {Phys.
  Rev. A}\ }\textbf {\bibinfo {volume} {84}},\ \bibinfo {pages} {032115}
  (\bibinfo {year} {2011})}\BibitemShut {NoStop}%
\bibitem [{\citenamefont {Schneeloch}\ \emph
  {et~al.}(2013{\natexlab{a}})\citenamefont {Schneeloch}, \citenamefont
  {Broadbent}, \citenamefont {Walborn}, \citenamefont {Cavalcanti},\ and\
  \citenamefont {Howell}}]{schneeloch_einstein-podolsky-rosen_2013}%
  \BibitemOpen
  \bibfield  {author} {\bibinfo {author} {\bibfnamefont {J.}~\bibnamefont
  {Schneeloch}}, \bibinfo {author} {\bibfnamefont {C.~J.}\ \bibnamefont
  {Broadbent}}, \bibinfo {author} {\bibfnamefont {S.~P.}\ \bibnamefont
  {Walborn}}, \bibinfo {author} {\bibfnamefont {E.~G.}\ \bibnamefont
  {Cavalcanti}}, \ and\ \bibinfo {author} {\bibfnamefont {J.~C.}\ \bibnamefont
  {Howell}},\ }\bibfield  {title} {\enquote {\bibinfo {title}
  {Einstein-{Podolsky}-{Rosen} steering inequalities from entropic uncertainty
  relations},}\ }\href {\doibase 10.1103/PhysRevA.87.062103} {\bibfield
  {journal} {\bibinfo  {journal} {Phys. Rev. A}\ }\textbf {\bibinfo {volume}
  {87}},\ \bibinfo {pages} {062103} (\bibinfo {year}
  {2013}{\natexlab{a}})}\BibitemShut {NoStop}%
\bibitem [{\citenamefont {Costa}\ and\ \citenamefont
  {Angelo}(2016)}]{costa_quantification_2016}%
  \BibitemOpen
  \bibfield  {author} {\bibinfo {author} {\bibfnamefont {A.~C.~S.}\
  \bibnamefont {Costa}}\ and\ \bibinfo {author} {\bibfnamefont {R.~M.}\
  \bibnamefont {Angelo}},\ }\bibfield  {title} {\enquote {\bibinfo {title}
  {Quantification of {Einstein}-{Podolski}-{Rosen} steering for two-qubit
  states},}\ }\href {\doibase 10.1103/PhysRevA.93.020103} {\bibfield  {journal}
  {\bibinfo  {journal} {Phys. Rev. A}\ }\textbf {\bibinfo {volume} {93}},\
  \bibinfo {pages} {020103} (\bibinfo {year} {2016})}\BibitemShut {NoStop}%
\bibitem [{\citenamefont {Zhu}\ \emph {et~al.}(2016)\citenamefont {Zhu},
  \citenamefont {Hayashi},\ and\ \citenamefont {Chen}}]{zhu_universal_2016}%
  \BibitemOpen
  \bibfield  {author} {\bibinfo {author} {\bibfnamefont {H.}~\bibnamefont
  {Zhu}}, \bibinfo {author} {\bibfnamefont {M.}~\bibnamefont {Hayashi}}, \ and\
  \bibinfo {author} {\bibfnamefont {L.}~\bibnamefont {Chen}},\ }\bibfield
  {title} {\enquote {\bibinfo {title} {Universal steering criteria},}\ }\href
  {\doibase 10.1103/PhysRevLett.116.070403} {\bibfield  {journal} {\bibinfo
  {journal} {Phys. Rev. Lett.}\ }\textbf {\bibinfo {volume} {116}},\ \bibinfo
  {pages} {070403} (\bibinfo {year} {2016})}\BibitemShut {NoStop}%
\bibitem [{\citenamefont {Saunders}\ \emph {et~al.}(2010)\citenamefont
  {Saunders}, \citenamefont {Jones}, \citenamefont {Wiseman},\ and\
  \citenamefont {Pryde}}]{saunders_experimental_2010}%
  \BibitemOpen
  \bibfield  {author} {\bibinfo {author} {\bibfnamefont {D.~J.}\ \bibnamefont
  {Saunders}}, \bibinfo {author} {\bibfnamefont {S.~J.}\ \bibnamefont {Jones}},
  \bibinfo {author} {\bibfnamefont {H.~M.}\ \bibnamefont {Wiseman}}, \ and\
  \bibinfo {author} {\bibfnamefont {G.~J.}\ \bibnamefont {Pryde}},\ }\bibfield
  {title} {\enquote {\bibinfo {title} {Experimental {EPR}-steering using
  {Bell}-local states},}\ }\href {\doibase 10.1038/nphys1766} {\bibfield
  {journal} {\bibinfo  {journal} {Nat. Phys.}\ }\textbf {\bibinfo {volume}
  {6}},\ \bibinfo {pages} {845--849} (\bibinfo {year} {2010})}\BibitemShut
  {NoStop}%
\bibitem [{\citenamefont {Schneeloch}\ \emph
  {et~al.}(2013{\natexlab{b}})\citenamefont {Schneeloch}, \citenamefont
  {Dixon}, \citenamefont {Howland}, \citenamefont {Broadbent},\ and\
  \citenamefont {Howell}}]{schneeloch_violation_2013}%
  \BibitemOpen
  \bibfield  {author} {\bibinfo {author} {\bibfnamefont {J.}~\bibnamefont
  {Schneeloch}}, \bibinfo {author} {\bibfnamefont {P.~B.}\ \bibnamefont
  {Dixon}}, \bibinfo {author} {\bibfnamefont {G.~A.}\ \bibnamefont {Howland}},
  \bibinfo {author} {\bibfnamefont {C.~J.}\ \bibnamefont {Broadbent}}, \ and\
  \bibinfo {author} {\bibfnamefont {J.~C.}\ \bibnamefont {Howell}},\ }\bibfield
   {title} {\enquote {\bibinfo {title} {Violation of continuous-variable
  {Einstein}-{Podolsky}-{Rosen} steering with discrete measurements},}\ }\href
  {\doibase 10.1103/PhysRevLett.110.130407} {\bibfield  {journal} {\bibinfo
  {journal} {Phys. Rev. Lett.}\ }\textbf {\bibinfo {volume} {110}},\ \bibinfo
  {pages} {130407} (\bibinfo {year} {2013}{\natexlab{b}})}\BibitemShut
  {NoStop}%
\bibitem [{\citenamefont {Piani}(2015)}]{piani_channel_2015}%
  \BibitemOpen
  \bibfield  {author} {\bibinfo {author} {\bibfnamefont {M.}~\bibnamefont
  {Piani}},\ }\bibfield  {title} {\enquote {\bibinfo {title} {Channel
  steering},}\ }\href {\doibase 10.1364/JOSAB.32.0000A1} {\bibfield  {journal}
  {\bibinfo  {journal} {J. Opt. Soc. Am. B}\ }\textbf {\bibinfo {volume}
  {32}},\ \bibinfo {pages} {A1--A7} (\bibinfo {year} {2015})}\BibitemShut
  {NoStop}%
\bibitem [{\citenamefont {Piani}\ and\ \citenamefont
  {Watrous}(2015)}]{piani_necessary_2015}%
  \BibitemOpen
  \bibfield  {author} {\bibinfo {author} {\bibfnamefont {M.}~\bibnamefont
  {Piani}}\ and\ \bibinfo {author} {\bibfnamefont {J.}~\bibnamefont
  {Watrous}},\ }\bibfield  {title} {\enquote {\bibinfo {title} {Necessary and
  sufficient quantum information characterization of
  {Einstein}-{Podolsky}-{Rosen} steering},}\ }\href {\doibase
  10.1103/PhysRevLett.114.060404} {\bibfield  {journal} {\bibinfo  {journal}
  {Phys. Rev. Lett.}\ }\textbf {\bibinfo {volume} {114}},\ \bibinfo {pages}
  {060404} (\bibinfo {year} {2015})}\BibitemShut {NoStop}%
\bibitem [{\citenamefont {Uola}\ \emph {et~al.}(2014)\citenamefont {Uola},
  \citenamefont {Moroder},\ and\ \citenamefont {G{\"u}hne}}]{uola_joint_2014}%
  \BibitemOpen
  \bibfield  {author} {\bibinfo {author} {\bibfnamefont {R.}~\bibnamefont
  {Uola}}, \bibinfo {author} {\bibfnamefont {T.}~\bibnamefont {Moroder}}, \
  and\ \bibinfo {author} {\bibfnamefont {O.}~\bibnamefont {G{\"u}hne}},\
  }\bibfield  {title} {\enquote {\bibinfo {title} {Joint measurability of
  generalized measurements implies classicality},}\ }\href {\doibase
  10.1103/PhysRevLett.113.160403} {\bibfield  {journal} {\bibinfo  {journal}
  {Phys. Rev. Lett.}\ }\textbf {\bibinfo {volume} {113}},\ \bibinfo {pages}
  {160403} (\bibinfo {year} {2014})}\BibitemShut {NoStop}%
\bibitem [{\citenamefont {Quintino}\ \emph {et~al.}(2014)\citenamefont
  {Quintino}, \citenamefont {V{\'e}rtesi},\ and\ \citenamefont
  {Brunner}}]{quintino_joint_2014}%
  \BibitemOpen
  \bibfield  {author} {\bibinfo {author} {\bibfnamefont {M.~T.}\ \bibnamefont
  {Quintino}}, \bibinfo {author} {\bibfnamefont {T.}~\bibnamefont
  {V{\'e}rtesi}}, \ and\ \bibinfo {author} {\bibfnamefont {N.}~\bibnamefont
  {Brunner}},\ }\bibfield  {title} {\enquote {\bibinfo {title} {Joint
  measurability, {Einstein}-{Podolsky}-{Rosen} steering, and {Bell}
  nonlocality},}\ }\href {\doibase 10.1103/PhysRevLett.113.160402} {\bibfield
  {journal} {\bibinfo  {journal} {Phys. Rev. Lett.}\ }\textbf {\bibinfo
  {volume} {113}},\ \bibinfo {pages} {160402} (\bibinfo {year}
  {2014})}\BibitemShut {NoStop}%
\bibitem [{\citenamefont {Uola}\ \emph {et~al.}(2015)\citenamefont {Uola},
  \citenamefont {Budroni}, \citenamefont {G{\"u}hne},\ and\ \citenamefont
  {Pellonp\"a\"a}}]{uola_one--one_2015}%
  \BibitemOpen
  \bibfield  {author} {\bibinfo {author} {\bibfnamefont {R.}~\bibnamefont
  {Uola}}, \bibinfo {author} {\bibfnamefont {C.}~\bibnamefont {Budroni}},
  \bibinfo {author} {\bibfnamefont {O.}~\bibnamefont {G{\"u}hne}}, \ and\
  \bibinfo {author} {\bibfnamefont {J.-P.}\ \bibnamefont {Pellonp\"a\"a}},\
  }\bibfield  {title} {\enquote {\bibinfo {title} {One-to-one mapping between
  steering and joint measurability problems},}\ }\href {\doibase
  10.1103/PhysRevLett.115.230402} {\bibfield  {journal} {\bibinfo  {journal}
  {Phys. Rev. Lett.}\ }\textbf {\bibinfo {volume} {115}},\ \bibinfo {pages}
  {230402} (\bibinfo {year} {2015})}\BibitemShut {NoStop}%
\bibitem [{\citenamefont {Uola}\ \emph {et~al.}(2016)\citenamefont {Uola},
  \citenamefont {Luoma}, \citenamefont {Moroder},\ and\ \citenamefont
  {Heinosaari}}]{uola_adaptive_2016}%
  \BibitemOpen
  \bibfield  {author} {\bibinfo {author} {\bibfnamefont {R.}~\bibnamefont
  {Uola}}, \bibinfo {author} {\bibfnamefont {K.}~\bibnamefont {Luoma}},
  \bibinfo {author} {\bibfnamefont {T.}~\bibnamefont {Moroder}}, \ and\
  \bibinfo {author} {\bibfnamefont {T.}~\bibnamefont {Heinosaari}},\ }\bibfield
   {title} {\enquote {\bibinfo {title} {Adaptive strategy for joint
  measurements},}\ }\href {\doibase 10.1103/PhysRevA.94.022109} {\bibfield
  {journal} {\bibinfo  {journal} {Phys. Rev. A}\ }\textbf {\bibinfo {volume}
  {94}},\ \bibinfo {pages} {022109} (\bibinfo {year} {2016})}\BibitemShut
  {NoStop}%
\bibitem [{\citenamefont {Nguyen}\ and\ \citenamefont
  {Vu}(2016{\natexlab{a}})}]{nguyen2016non}%
  \BibitemOpen
  \bibfield  {author} {\bibinfo {author} {\bibfnamefont {H.~C.}\ \bibnamefont
  {Nguyen}}\ and\ \bibinfo {author} {\bibfnamefont {T.}~\bibnamefont {Vu}},\
  }\bibfield  {title} {\enquote {\bibinfo {title} {Nonseparability and
  steerability of two-qubit states from the geometry of steering outcomes},}\
  }\href {\doibase 10.1103/PhysRevA.94.012114} {\bibfield  {journal} {\bibinfo
  {journal} {Phys. Rev. A}\ }\textbf {\bibinfo {volume} {94}},\ \bibinfo
  {pages} {012114} (\bibinfo {year} {2016}{\natexlab{a}})}\BibitemShut
  {NoStop}%
\bibitem [{\citenamefont {Jevtic}\ \emph {et~al.}(2014)\citenamefont {Jevtic},
  \citenamefont {Pusey}, \citenamefont {Jennings},\ and\ \citenamefont
  {Rudolph}}]{jevtic_quantum_2014}%
  \BibitemOpen
  \bibfield  {author} {\bibinfo {author} {\bibfnamefont {S.}~\bibnamefont
  {Jevtic}}, \bibinfo {author} {\bibfnamefont {M.}~\bibnamefont {Pusey}},
  \bibinfo {author} {\bibfnamefont {D.}~\bibnamefont {Jennings}}, \ and\
  \bibinfo {author} {\bibfnamefont {T.}~\bibnamefont {Rudolph}},\ }\bibfield
  {title} {\enquote {\bibinfo {title} {Quantum steering {Ellipsoids}},}\ }\href
  {\doibase 10.1103/PhysRevLett.113.020402} {\bibfield  {journal} {\bibinfo
  {journal} {Phys. Rev. Lett.}\ }\textbf {\bibinfo {volume} {113}},\ \bibinfo
  {pages} {020402} (\bibinfo {year} {2014})}\BibitemShut {NoStop}%
\bibitem [{\citenamefont {Jevtic}\ \emph {et~al.}(2015)\citenamefont {Jevtic},
  \citenamefont {Hall}, \citenamefont {Anderson}, \citenamefont {Zwierz},\ and\
  \citenamefont {Wiseman}}]{jevtic2015einstein}%
  \BibitemOpen
  \bibfield  {author} {\bibinfo {author} {\bibfnamefont {S.}~\bibnamefont
  {Jevtic}}, \bibinfo {author} {\bibfnamefont {M.~J.~W.}\ \bibnamefont {Hall}},
  \bibinfo {author} {\bibfnamefont {M.~R.}\ \bibnamefont {Anderson}}, \bibinfo
  {author} {\bibfnamefont {M.}~\bibnamefont {Zwierz}}, \ and\ \bibinfo {author}
  {\bibfnamefont {H.~M.}\ \bibnamefont {Wiseman}},\ }\bibfield  {title}
  {\enquote {\bibinfo {title} {Einstein--podolsky--rosen steering and the
  steering ellipsoid},}\ }\href@noop {} {\bibfield  {journal} {\bibinfo
  {journal} {J. Opt. Soc. Am. B}\ }\textbf {\bibinfo {volume} {32}},\ \bibinfo
  {pages} {A40--A49} (\bibinfo {year} {2015})}\BibitemShut {NoStop}%
\bibitem [{\citenamefont {Nguyen}\ and\ \citenamefont
  {Vu}(2016{\natexlab{b}})}]{nguyen2016necessary}%
  \BibitemOpen
  \bibfield  {author} {\bibinfo {author} {\bibfnamefont {H.~Chau}\ \bibnamefont
  {Nguyen}}\ and\ \bibinfo {author} {\bibfnamefont {T.}~\bibnamefont {Vu}},\
  }\bibfield  {title} {\enquote {\bibinfo {title} {Necessary and sufficient
  condition for steerability of two-qubit states by the geometry of steering
  outcomes},}\ }\href {http://stacks.iop.org/0295-5075/115/i=1/a=10003}
  {\bibfield  {journal} {\bibinfo  {journal} {Europhys. Lett.}\ }\textbf
  {\bibinfo {volume} {115}},\ \bibinfo {pages} {10003} (\bibinfo {year}
  {2016}{\natexlab{b}})}\BibitemShut {NoStop}%
\bibitem [{\citenamefont {Milne}\ \emph {et~al.}(2014)\citenamefont {Milne},
  \citenamefont {Jevtic}, \citenamefont {Jennings}, \citenamefont {Wiseman},\
  and\ \citenamefont {Rudolph}}]{milne2014quantum}%
  \BibitemOpen
  \bibfield  {author} {\bibinfo {author} {\bibfnamefont {A.}~\bibnamefont
  {Milne}}, \bibinfo {author} {\bibfnamefont {S.}~\bibnamefont {Jevtic}},
  \bibinfo {author} {\bibfnamefont {D.}~\bibnamefont {Jennings}}, \bibinfo
  {author} {\bibfnamefont {H.}~\bibnamefont {Wiseman}}, \ and\ \bibinfo
  {author} {\bibfnamefont {T.}~\bibnamefont {Rudolph}},\ }\bibfield  {title}
  {\enquote {\bibinfo {title} {Quantum steering ellipsoids, extremal physical
  states and monogamy},}\ }\href
  {http://iopscience.iop.org/article/10.1088/1367-2630/16/8/083017/meta}
  {\bibfield  {journal} {\bibinfo  {journal} {New J. Phys.}\ }\textbf {\bibinfo
  {volume} {16}},\ \bibinfo {pages} {083017} (\bibinfo {year}
  {2014})}\BibitemShut {NoStop}%
\bibitem [{\citenamefont {Milne}\ \emph {et~al.}(2015)\citenamefont {Milne},
  \citenamefont {Jennings},\ and\ \citenamefont
  {Rudolph}}]{milne2015geometric}%
  \BibitemOpen
  \bibfield  {author} {\bibinfo {author} {\bibfnamefont {A.}~\bibnamefont
  {Milne}}, \bibinfo {author} {\bibfnamefont {D.}~\bibnamefont {Jennings}}, \
  and\ \bibinfo {author} {\bibfnamefont {T.}~\bibnamefont {Rudolph}},\
  }\bibfield  {title} {\enquote {\bibinfo {title} {Geometric representation of
  two-qubit entanglement witnesses},}\ }\href
  {http://journals.aps.org/pra/abstract/10.1103/PhysRevA.92.012311} {\bibfield
  {journal} {\bibinfo  {journal} {Phys. Rev. A}\ }\textbf {\bibinfo {volume}
  {92}},\ \bibinfo {pages} {012311} (\bibinfo {year} {2015})}\BibitemShut
  {NoStop}%
\bibitem [{\citenamefont {Chen}\ \emph {et~al.}(2013)\citenamefont {Chen},
  \citenamefont {Ye}, \citenamefont {Wu}, \citenamefont {Su}, \citenamefont
  {Cabello}, \citenamefont {Kwek},\ and\ \citenamefont
  {Oh}}]{chen_all-versus-nothing_2013}%
  \BibitemOpen
  \bibfield  {author} {\bibinfo {author} {\bibfnamefont {J.-L.}\ \bibnamefont
  {Chen}}, \bibinfo {author} {\bibfnamefont {X.-J.}\ \bibnamefont {Ye}},
  \bibinfo {author} {\bibfnamefont {C.}~\bibnamefont {Wu}}, \bibinfo {author}
  {\bibfnamefont {H.-Y.}\ \bibnamefont {Su}}, \bibinfo {author} {\bibfnamefont
  {A.}~\bibnamefont {Cabello}}, \bibinfo {author} {\bibfnamefont {L.~C.}\
  \bibnamefont {Kwek}}, \ and\ \bibinfo {author} {\bibfnamefont {C.~H.}\
  \bibnamefont {Oh}},\ }\bibfield  {title} {\enquote {\bibinfo {title}
  {All-versus-nothing proof of {Einstein}-{Podolsky}-{Rosen} steering},}\
  }\href {\doibase 10.1038/srep02143} {\bibfield  {journal} {\bibinfo
  {journal} {Sci. Rep.}\ }\textbf {\bibinfo {volume} {3}} (\bibinfo {year}
  {2013}),\ 10.1038/srep02143}\BibitemShut {NoStop}%
\bibitem [{\citenamefont {Heinosaari}\ and\ \citenamefont
  {Ziman}(2011)}]{heinosaari_mathematical_2011}%
  \BibitemOpen
  \bibfield  {author} {\bibinfo {author} {\bibfnamefont {T.}~\bibnamefont
  {Heinosaari}}\ and\ \bibinfo {author} {\bibfnamefont {M.}~\bibnamefont
  {Ziman}},\ }\href {http://books.google.fi/books?id=cZ8vd1nTI0EC} {\emph
  {\bibinfo {title} {The Mathematical Language of Quantum Theory: From
  Uncertainty to Entanglement}}}\ (\bibinfo  {publisher} {Cambridge University
  Press},\ \bibinfo {year} {2011})\BibitemShut {NoStop}%
\bibitem [{\citenamefont {He}\ \emph {et~al.}(2011)\citenamefont {He},
  \citenamefont {Drummond},\ and\ \citenamefont {Reid}}]{he_entanglement_2011}%
  \BibitemOpen
  \bibfield  {author} {\bibinfo {author} {\bibfnamefont {Q.~Y.}\ \bibnamefont
  {He}}, \bibinfo {author} {\bibfnamefont {P.~D.}\ \bibnamefont {Drummond}}, \
  and\ \bibinfo {author} {\bibfnamefont {M.~D.}\ \bibnamefont {Reid}},\
  }\bibfield  {title} {\enquote {\bibinfo {title} {Entanglement, {EPR}
  steering, and {Bell}-nonlocality criteria for multipartite higher-spin
  systems},}\ }\href {\doibase 10.1103/PhysRevA.83.032120} {\bibfield
  {journal} {\bibinfo  {journal} {Phys. Rev. A}\ }\textbf {\bibinfo {volume}
  {83}},\ \bibinfo {pages} {032120} (\bibinfo {year} {2011})}\BibitemShut
  {NoStop}%
\bibitem [{\citenamefont {Sun}\ \emph {et~al.}(2014)\citenamefont {Sun},
  \citenamefont {Xu}, \citenamefont {Ye}, \citenamefont {Wu}, \citenamefont
  {Chen}, \citenamefont {Li},\ and\ \citenamefont
  {Guo}}]{sun_experimental_2014}%
  \BibitemOpen
  \bibfield  {author} {\bibinfo {author} {\bibfnamefont {K.}~\bibnamefont
  {Sun}}, \bibinfo {author} {\bibfnamefont {J.-S.}\ \bibnamefont {Xu}},
  \bibinfo {author} {\bibfnamefont {X.-J.}\ \bibnamefont {Ye}}, \bibinfo
  {author} {\bibfnamefont {Y.-C.}\ \bibnamefont {Wu}}, \bibinfo {author}
  {\bibfnamefont {J.-L.}\ \bibnamefont {Chen}}, \bibinfo {author}
  {\bibfnamefont {C.-F.}\ \bibnamefont {Li}}, \ and\ \bibinfo {author}
  {\bibfnamefont {G.-C.}\ \bibnamefont {Guo}},\ }\bibfield  {title} {\enquote
  {\bibinfo {title} {Experimental demonstration of the
  {Einstein}-{Podolsky}-{Rosen} steering game based on the all-versus-nothing
  proof},}\ }\href {\doibase 10.1103/PhysRevLett.113.140402} {\bibfield
  {journal} {\bibinfo  {journal} {Phys. Rev. Lett.}\ }\textbf {\bibinfo
  {volume} {113}},\ \bibinfo {pages} {140402} (\bibinfo {year}
  {2014})}\BibitemShut {NoStop}%
\end{thebibliography}%
\end{document}